\documentclass[preprint]{sig-alternate}

\conferenceinfo{SAC'14,} {Mar. 24-28, 2014. Gyeongju, Korea}
\CopyrightYear{2014} 
\crdata{TBD} 
\usepackage{comment,amsmath,amssymb,verbatim,url,bm}

\newtheorem{definition}{Definition}[section]
\newtheorem{theorem}{Theorem}[section]

\newtheorem{con}{Construction}[section]
\date{\today}

\setpagenumber{1}
\begin{document}

\toappear{Permission to make digital or hard copies of all or part of this work for personal or classroom use is granted without fee provided that copies are
not made or distributed for profit or commercial advantage and that copies
bear this notice and the full citation on the first page. To copy otherwise, to
republish, to post on servers or to redistribute to lists, requires prior specific permission and/or a fee.\\
This is a preprint version of our paper presented in SAC'14, March 24-28, 2014, Gyeongju, Korea.\\}

\title{Building Secure and Anonymous Communication Channel: Formal Model and its Prototype Implementation}

\author{
\numberofauthors{2}
% 1st. author
\alignauthor
Keita Emura\\
       \affaddr{National Institute of Information and Communications Technology, Japan}\\
       \email{k-emura@nict.go.jp}
% 2nd. author
\alignauthor
Akira Kanaoka\\
       \affaddr{National Institute of Information and Communications Technology, and \\Toho University, Japan}\\
       \email{akira.kanaoka@is.sci.toho-u.ac.jp}
% 3rd. author
\and 
\alignauthor 
Satoshi Ohta\\
       \affaddr{National Institute of Information and Communications Technology, Japan}\\
       \email{sota@nict.go.jp}
%\and  % use '\and' if you need 'another row' of author names
% 4th. author
\alignauthor Takeshi Takahashi\\
       \affaddr{National Institute of Information and Communications Technology, Japan}\\
       \email{takeshi\_takahashi@ieee.org}
}

\maketitle
\begin{abstract}
Various techniques need to be combined to realize anonymously authenticated communication.
Cryptographic tools enable anonymous user authentication while anonymous communication protocols hide users' IP addresses from service providers.
One simple approach for realizing anonymously authenticated communication is their simple combination, but this gives rise to another issue; how to build a secure channel.
The current public key infrastructure cannot be used since the user's public key identifies the user.
To cope with this issue, we propose a protocol that uses identity-based encryption for packet encryption without sacrificing anonymity, and group signature for anonymous user authentication.
Communications in the protocol take place through proxy entities that conceal users' IP addresses from service providers.
The underlying group signature is customized to meet our objective and improve its efficiency.
We also introduce a proof-of-concept implementation to demonstrate the protocol's feasibility.
We compare its performance to SSL communication and demonstrate its practicality, and conclude that the protocol realizes secure, anonymous, and authenticated communication between users and service providers with practical performance.
\end{abstract}

%%%%%%%%%%%%%%%%%%%%%%%%%%%%%%%%%%%%%%%%
% ACM specific paper category information
%%%%%%%%%%%%%%%%%%%%%%%%%%%%%%%%%%%%%%%%

\category{C.2.0}{Computer-Communication Networks}{General}[Security and Protection]
\category{H.1.m}{Information System}{Models and Principles}[Miscellaneous]
\category{E.3}{Data}{Data Encryption}[Public key cryptosystems]
%\category{K.6.5}{Management of Computing and Information Systems}{Security and Protection}

\terms{Security, Design, Theory}

%%%%%%%%%%%%%%%%%%%%%%%%%%%%%%%%%%%%%%%%
% Keywords
%%%%%%%%%%%%%%%%%%%%%%%%%%%%%%%%%%%%%%%%

\keywords{Anonymous Communication, Anonymous Authentication, Secure Channel, Identity-Based Encryption, Group Signature}

%%%%%%%%%%%%%%%%%%%%%%%%%%%%%%%%%%%%%%%%
% Introduction
%%%%%%%%%%%%%%%%%%%%%%%%%%%%%%%%%%%%%%%%

\section{Introduction}

Anonymity\footnote{This paper considers sender/prover anonymity and does not consider recipient anonymity.} is an important aspect of privacy, and systems that provide services to anonymous users are currently a topic of keen interest. 
Such systems can provide services to users without revealing their identity. 
To date, a great deal of studies have been reported on~\cite{SPA}, 
and many use cryptography as the important building block for constructing the systems, 
but these need further improvement before they can be used for actual services.
To realize secure, anonymous, and authenticated communication, these building blocks need to collaborate with each other.

\subsection{Research Background}

Several cryptographic primitives providing anonymity have been proposed. 
Among them is group signature~\cite{[ChaumH91]}, which enables a signer to anonymously prove signatures' validity.
A group manager (GM) that has a pair of a group public key, $gpk$, and master secret key, $msk$, issues a secret signing key, $sk_i$, to a user $U_i$, which computes a group signature, $\sigma$ (on certain messages), using $sk_i$. 
No user-dependent value is required in the verification phase; a verifier verifies $\sigma$ using only the corresponding $gpk$. 
%and judges whether a signer is a group member without identifying the signer.
These approaches alone, however, cannot guarantee anonymity when applied to online communication.
For instance, let a signer compute a group signature and \emph{send} it to a verifier. 
The verifier can anonymously verify the signature's validity. 
%but it cannot know the signer's identity due to the group signature.
However, there is a question of how to anonymously send the group signature to the verifier. 
Usually, a source IP address is included in a packet, and that reveals the identity of the sender, thus user anonymity is already infringed upon.
The situation remains the same regardless of the primitives we implement so long as direct communication between a sender (signer, prover, etc) and a receiver (verifier, etc) is required.

The user's IP address is naturally visible in the IP packets sent from the user, and it cannot simply be erased or forged to enable bi-directional communication. 
One approach for this is using intermediate agents that send packets on behalf of the actual user terminal, and several such protocols have already been proposed~\cite{SPA}, including Tor~\cite{TorProject}. 
%If an anonymous channel can be established between end users and a server by using intermediate agents, the source IP can be hidden.
Nevertheless, another issue arises in the question of how we can assure user's legitimacy. 
We need to discern legitimate and illegitimate users to restrict unauthorized access to the channel.
One might think that only end-to-end authentication is needed, but it is hard to authenticate users without identifying them.
For instance, a server needs to send a response code to a user in basic authentication and the user needs to return a user ID and password.
That is, the server needs to identify the user.
Moreover, it seems to be hard to send a certain message back from the server to a user since the corresponding source IP address is generally required. 
Authentication by an intermediate agent (as in Tor~\cite{TorProject}) might be a solution to these problems.
The agent can authenticate a user and can hide the user's source IP address from the server. 
That notwithstanding, we still need to know how the server can directly authenticate end users. 
%Indeed, users are authenticated when they log in a system, and it is unrealistic to assume that all such users are legitimate.
%Moreover, we cannot assume that all the intermediate agents are actually trustworthy.

A simple approach to the anonymous authentication problems is just combining both cryptographic primitives and anonymous communication protocols as follows. 
Let a user compute an anonymously-authenticated token (e.g., group signature), 
and send it to a server via an anonymous channel (e.g., using Tor). 
Then, the server can directly authenticate the user without compromising anonymity. 
However, another problem arises is how we can establish a secure channel (i.e., flowed data is encrypted). 
If the server uses a user's public key (certified by a trusted Certificate Authority (CA) in a public key infrastructure (PKI)), 
then server identifies the user, since a certificate contains information on the key holder. 
The same problem arises even if symmetric key encryption is used. 
Assume that the server tries to exchange a secret key with an end user.
Since the server does not know who the actual end user is, the server does not know the user's public key for running a key exchange protocol. 
%To sum up, it is hard to realize secure, anonymous, and authenticated communication just by combining anonymous communication protocols, anonymous authentication protocols, and other cryptographic tools

\subsection{Our Contribution}

We propose a protocol that realizes secure, anonymous, and authenticated communication.
The proposed protocol uses identity-based encryption (IBE) to encrypt content without identifying key holders\footnote{The conventional public key encryption (PKE) with certain non-interactive zero-knowledge (NIZK) proofs may also be applicable, where a user chooses a temporary public key for each session, and makes a NIZK proof of the corresponding secret key. We do not consider this construction anymore since the NIZK proof must be constructed from scratch by considering algebraic structures of the underlying PKE scheme, and this may lead to some difficulty of its implementation. }.
It can set arbitrary values on public keys, thus it can enable a user to select a temporary ID for each session, which the server can use as a public key.
It also uses group signature for anonymous user authentication. 
Communications in the protocol take place through proxy entities that conceal users' IP addresses from service providers (SPs). 

This paper gives the framework of the proposed protocol, gives a formal model and security definitions of the proposed protocol, points out the needlessness of the group signature's open capability for our usage, 
%, where only GM can identify signatures' signers, 
and then proposes an open-free variant of the Furukawa-Imai group signature scheme~\cite{[FurukawaI06]}. 
The modification can reduce its signature size by 50\% compared to the original scheme.
Note that if someone needs to identify an illegitimate user, we can add such a mechanism without relying on cryptographic techniques; e.g., an IP address managed by the proxy. 

We demonstrate the feasibility and practicality of the proposed protocol by introducing our proof-of-concept implementation.
The implementation uses the modified group signature scheme mentioned above.
It also uses the Boneh-Franklin IBE scheme~\cite{[BonehF03]} for its underlying IBE scheme. 

Note that, our protocol in this paper focuses on encrypted communication from the SP to users. 
%\footnote{One use case of this type of application is push-type secure information delivery service.}.
It can easily be extended to interactive secure communication since SP is not anonymous to users and each user thus can simply use SP's public key for building a secure channel.

\subsection{Related Work}

There exist similar attempts to our approach.
Sudarsono et al.~\cite{[SudarsonoNNF10]} has considered an anonymous IEEE802.1X authentication system by using a group signature scheme.
They use group signatures as the client's digital certificate.
The means of sending such certificates over IP networks was, however, outside its scope. 
Lee et al.~\cite{anon-pass13} proposed an anonymous subscription service, called Anon-pass.
Their construction methodology is similar to group signatures, wherein a user proves the possession of signatures using zero-knowledge proofs.
Though Anon-pass does not consider end-to-end secure (encrypted) communication, our protocol does. 

Gilad and Herzberg~\cite{[GiladH13]} also considered how to distribute public keys using an anonymous service.
They consider two peers, a querier and a responder.
The querier specifies a random ephemeral public key that is not certified by the CA, 
and sends a query containing this public key to a responder via an anonymous service, like Tor. 
The responder replies with a response message encrypted by the (anonymous) querier's ephemeral public key.
However, a responder cannot check whether a public key is a valid key or a random value since this scheme gives no certification of the public key, and moreover the responder cannot detect even if the pubic key is replaced by an attacker. 
Moreover, no anonymous user authentication is considered in the Gilad-Herzberg system. 
In our protocol, the SP can be convinced that a public key (i.e., a temporary ID) will work, since arbitrary values can be public keys in IBE systems. 
Moreover, since a temporary ID is signed by group signature, 
we can prevent the key replacement attack and can achieve anonymous user authentication, simultaneously. 

Proxy re-encryption (PRE) (e.g.,~\cite{[LibertV11]}) 
is another candidate. 
%for building secure channels without conflicting anonymity. 
In our context, first users compute re-encryption keys using their secret key and the SP public key, and the SP only computes ciphertext using its own public key. Then, the proxy can re-encrypt ciphertexts. 
However, the proxy needs to manage all re-encryption keys, and therefore it is difficult to assume that no private information is infringed on even if the proxy is corrupted after the communication. Moreover, there is a possibility that other user may decrypt unexpected ciphertexts, since the proxy manages many re-encryption keys (from the SP to each user). 
%, and the original ciphertext is computed by the SP public key regardless of the recipients. 
It is undesirable to generate re-encryption keys that can be used in an unexpected manner, even if the proxy is modeled as an honest-but-curious entity and always follows the protocol. 
In our protocol no unexpected user (including the proxy) can decrypt ciphertexts, since a unique temporary ID is assigned for each user and each session. Note that the key escrow problem happens as an outcome of IBE, where key generation center (KGC) can decrypt all ciphertexts. However, KGC is modeled as a trusted third party, whereas it is difficult to fully trust all proxies involved in the systems. 

%\subsection{Organization} 
%This paper is organized as follows. We introduce the formal security model of IBE in Section 2. 
%Our open-free variant Furukawa-Imai group signature scheme is given in Section 3. 
%We give our proposed secure anonymous authentication protocol in Section 4. 
%An implementation result of our protocol (based on the Boneh-Franklin IBE scheme and our open-free group signature scheme) is shown in Section 5. 

%%%%%%%%%%%%%%%%%%%%%%%%%%%%%%%%%%%%%%%%
% Framework
%%%%%%%%%%%%%%%%%%%%%%%%%%%%%%%%%%%%%%%%

\section{Framework}

Figure \ref{fig1.eps} depicts the framework of the proposed protocol, 
which has five roles; User, Proxy, Service Provider, GM, and KGC.
A {\bf User} wishes to communicate with an SP without revealing its identity.
The {\bf Proxy} assists communication between a User and SP by relaying packets without revealing the User's IP address.
We assume that it is honest-but-curious.
The {\bf SP} provides services to Users, but wishes to authenticate them.
It does not care about their identities but needs to confirm whether the user accessing it is legitimate.
The {\bf GM} manages a group key and issuer key, and issues a signing key for a user that is used for generating an anonymously-authenticated token.
We assume that the GM suitably authenticates a user before issuing the signing key. 
The {\bf KGC} generates a decryption key for a user. 
We assume that the KGC suitably authenticates a user before issuing the decryption key. 

\begin{figure}[t]
\centering
{\includegraphics[scale=.31]{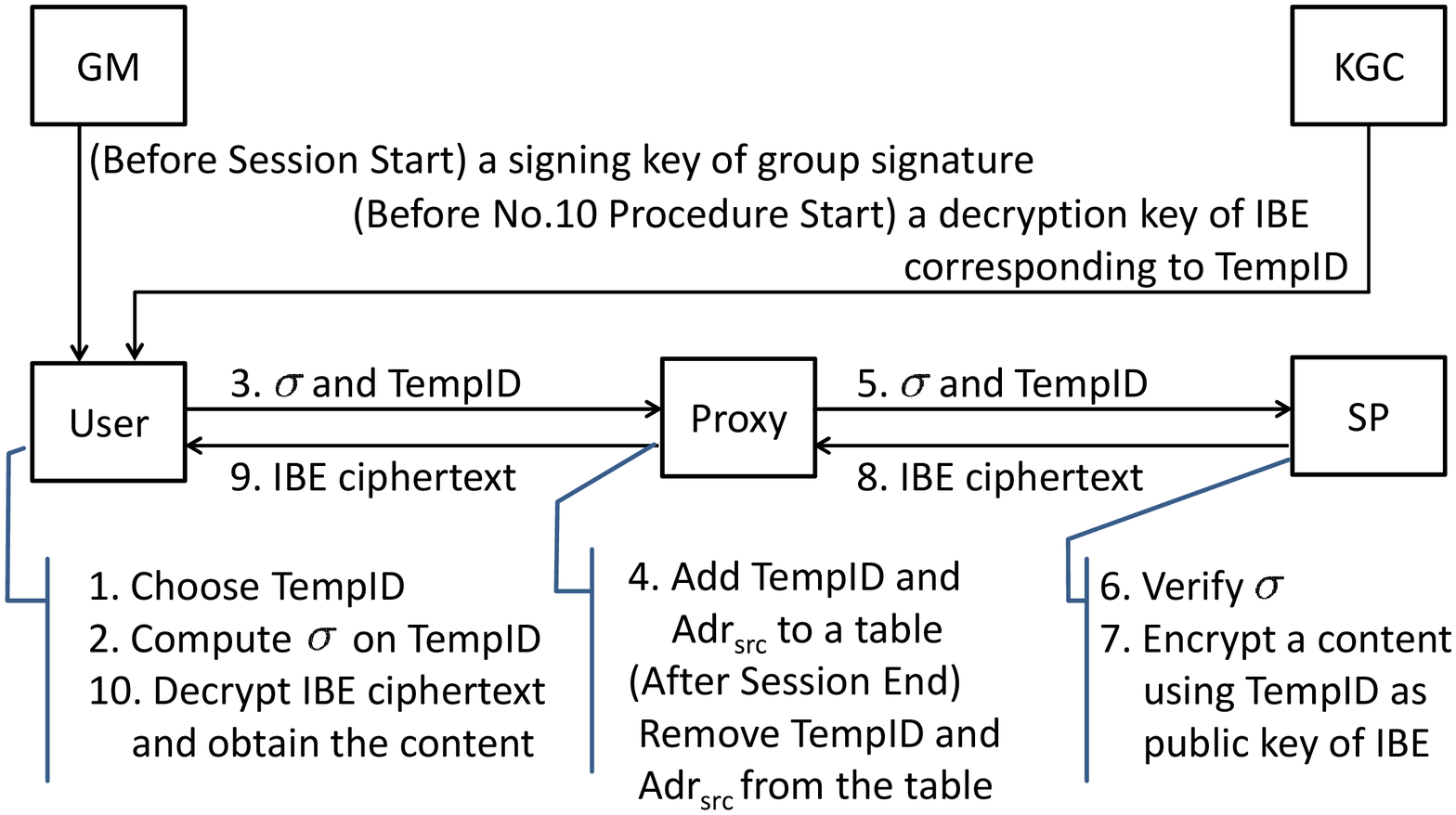}}
\caption{Framework of the Proposed Protocol}\label{fig1.eps}
\end{figure}

These roles need to collaborate with each other to realize the proposed secure anonymous authentication.
Their interaction sequence is as follows.
(1) A user (whose IP address is $\mathtt{Adr_{src}}$) chooses a temporary ID $\mathtt{TempID}$, (2) computes a group signature $\sigma$ on $\mathtt{TempID}$, and (3) sends $(\sigma, \mathtt{TempID})$ to the proxy.
(4) The proxy associates $\mathtt{Adr_{src}}$ with this temporary ID, and (5) forwards $(\sigma, \mathtt{TempID})$ to the SP.
(6) The SP can directly authenticate the users by verifying the group signature without compromising anonymous communications.
(7) If the user is successfully verified, the SP encrypts content using TempID as the public key of IBE; otherwise it returns $\bot$.
Here, we apply an IBE's property to establish a secure channel between the SP and an anonymous user, where arbitrary values can be a public key, and a ciphertext can be independently computed with the generation of the corresponding decryption key\footnote{This property is used in timed-release encryption~\cite{[CheonHKO08]} context, where an encryptor can control when ciphertexts will be decrypted.}. 
(8) The SP sends this IBE ciphertext to the proxy, which again (9) forwards it to the corresponding user. (10) Finally, the user decrypts the IBE ciphertext using the corresponding decryption key issued by the KGC. 
After relaying the (mutual) communication, the proxy immediately deletes the corresponding pair of $(\mathtt{TempID}, \mathtt{Adr_{src}})$. 
Therefore, no private information is infringed on even if the proxy is corrupted after the communication. 
Note that Figure~\ref{fig1.eps} explains one-pass communication.
The proxy can reuse the information of the pair $(\mathtt{TempID}, \mathtt{Adr_{src}})$ of a session so long as the session is alive, but it removes the information from its registry once the session is closed. 
%The proxy can reuse $(\mathtt{TempID}, \mathtt{Adr_{src}})$ during the session if more than one-pass communication is needed.
In either case, the proxy immediately deletes the corresponding pair $(\mathtt{TempID}, \mathtt{Adr_{src}})$ after the session.
Moreover, we can easily extend one-proxy setting to multi-proxy setting, since all the proxy has to do is (1) manage $(\mathtt{TempID}, \mathtt{Adr_{src}})$, and (2) forwarding $(\sigma, \mathtt{TempID})$ to the next. 
The above framework is embodied as a concrete protocol in the following section.

%%%%%%%%%%%%%%%%%%%%%%%%%%%%%%%%%%%%%%%%
% Proposed protocol
%%%%%%%%%%%%%%%%%%%%%%%%%%%%%%%%%%%%%%%%

\section{Authentication Protocol}

This section proposes a secure anonymous authentication protocol.\footnote{Note that our protocol achieves to send an anonymous token (group signature) and it is not an authentication protocol in the strict sense. Nevertheless, we can easily extend it to an authentication protocol (via the classical challenge-and-response methodology) as follows: first a SP sends a random nonce to a User (via the Proxy) and the User computes a group signature whose signed message contains the nonce. Therefore, we do not further consider the extension in this paper.}
It first defines the syntax of the protocol, 
and then gives its construction. 
We consider a scenario in which an SP is modeled as a server, which provides a service only for a legitimate user. 
That is, we can assume that the GM has authenticated a user before issuing a signing key, and the SP can judge that the user who can generate a valid group signature is a legitimate user. 

\subsection{Syntax and Security Definitions}

Let $\mathcal{ID}$ and $\mathcal{M}$ be an identity space and message space, respectively, 
and $\mathtt{Adr_{src}}$, $\mathtt{Adr_{proxy}}$, and $\mathtt{Adr_{dst}}$ stand for IP address of User, Proxy, and SP, respectively. 
For a set $X$ and an element $x\in X$, $x\stackrel{\$}{\leftarrow}X$ means that $x$ is randomly chosen from $X$. 

\begin{definition}[Syntax of The Protocol]~
\begin{description}
\setlength{\itemsep}{0em}\setlength{\parsep}{0em}
\item[${\sf GM.Setup}$]: This probabilistic algorithm takes as input the security parameter $\lambda$, and outputs a group public key $gpk$ and an issuer key $ik$.
\item[${\sf KGC.Setup}$]: This probabilistic algorithm takes as input the security parameter $\lambda$, and outputs a public key $params$ and a master secret key $msk$.  
\item[${\sf Join}$]: This probabilistic algorithm takes as input $gpk$ and $ik$, and outputs a signing key $sk$. 
\item[${\sf UserKeyGen}$]: This (possibly) probabilistic algorithm takes as input $params$, $msk$, and an (possibly temporary) identity $\mathtt{TempID}\in\mathcal{ID}$, and outputs a decryption key $dk_{\mathtt{TempID}}$. 
\item[${\sf SendRequest}$]: This probabilistic algorithm takes as input $gpk$, $sk$, $\mathtt{TempID}$, a source IP address $\mathtt{Adr_{src}}$, a destination IP address $\mathtt{Adr_{dst}}$, and a proxy IP address $\mathtt{Adr_{proxy}}$, and send a token $\sigma$, $\mathtt{TempID}$ and $\mathtt{Adr_{dst}}$ to the proxy whose IP address is $\mathtt{Adr_{proxy}}$. 
%We denote it $$(\sigma,\mathtt{TempID},\mathtt{Adr_{dst}})\stackrel{\mathtt{Adr_{proxy}}}{\longleftarrow}{\sf SendRequest}(gpk,sk,\mathtt{TempID},\mathtt{Adr_{src}},\mathtt{M}).$$ 
\item[${\sf RelayRequest}$]: This deterministic algorithm takes as input $\mathtt{Adr_{src}}$, $\mathtt{Adr_{dst}}$, an ID/IP table $\mathtt{Tbl}$, $\sigma$, and $\mathtt{TempID}$, and relays a pair $(\sigma,\mathtt{TempID})$ and $\mathtt{Adr_{proxy}}$ to the destination SP whose IP address is $\mathtt{Adr_{dst}}$. Moreover, append $(\mathtt{TempID}, \mathtt{Adr_{src}})$ to $\mathtt{Tbl}$.% We denote it $$(\sigma,\mathtt{TempID},\mathtt{Adr_{proxy}})\stackrel{\mathtt{Adr_{dst}}}{\longleftarrow}{\sf RelayRequest}(\mathtt{Adr_{src}},\mathtt{Adr_{dst}},\mathtt{Tbl},\sigma, \mathtt{TempID}).$$

\item[${\sf ValidityCheck}$]: This deterministic algorithm takes as input $gpk$, $\sigma$, and $\mathtt{TempID}$, and outputs 1 if $\sigma$ is valid, and 0, otherwise. 

\item[${\sf SendContent}$]: This probabilistic algorithm takes as input $gpk$, $\sigma$, $\mathtt{TempID}$, a content to be sent $M\in\mathcal{M}$, and $\mathtt{Adr_{proxy}}$, computes a ciphertext $C$ if the token $\sigma$ is valid, and sends $C$ to the proxy whose IP address is $\mathtt{Adr_{proxy}}$. Otherwise, if $\sigma$ is invalid, then return $\bot$. 
%We denote it $$C/\bot\stackrel{\mathtt{Adr_{proxy}}}{\longleftarrow}{\sf SendContent}(gpk,\sigma,\mathtt{TempID},M,\mathtt{Adr_{proxy}}).$$ 

\item[${\sf RelayContent}$]: This deterministic algorithm takes as input $C$ and $\mathtt{Tbl}$, and relays $C$ to a user whose IP address is $\mathtt{Adr_{src}}$ contained in $\mathtt{Tbl}$. Moreover, remove $(\mathtt{TempID}, \mathtt{Adr_{src}})$ from $\mathtt{Tbl}$. 
%We denote it $$C\stackrel{\mathtt{Adr_{src}}}{\longleftarrow}{\sf RelayContent}(C,\allowbreak \mathtt{Tbl}).$$
%
We assume that the proxy can decide the corresponding source IP address to be relayed by $C$\footnote{For example, port numbers can be used for identifying the sessions. It's up to the proxy in our implementation.}. 

\item[${\sf GetContent}$]: This deterministic algorithm takes as input $C$ and $dk_{\mathtt{TempID}}$, and return $M$. 
\end{description}
\end{definition}

Next, we give formal security definitions as follows. 
First, we define correctness that guarantees $\sigma$ is valid and a user always can obtain the corresponding content if all values are honestly generated according to the algorithms.

\begin{definition}[Correctness]
For all $(gpk,ik)\leftarrow {\sf GM.Setup}\allowbreak(1^\lambda)$, 
$(params,msk)\leftarrow {\sf KGC.Setup}(1^\lambda)$, 
$sk\leftarrow \linebreak {\sf Join}(gpk,ik)$,\allowbreak $\mathtt{TempID}\in\mathcal{ID}$, $M\in\mathcal{M}$, and $(\mathtt{Adr_{src}},\mathtt{Adr_{dst}},\allowbreak \mathtt{Adr_{proxy}})$, 

\begin{eqnarray*}
\Pr[M\leftarrow&&{\sf GetContent}\Big{(}{\sf RelayContent}\big{(}C,\mathtt{Tbl}\big{)},\\
&&{\sf UserKeyGen}(param, msk, \mathtt{TempID})\Big{)}]=1,~\text{and}\\
\Pr[1\leftarrow&&{\sf ValidityCheck}(gpk,\sigma,\mathtt{TempID})=1]=1
\end{eqnarray*}
\noindent 
where $(\sigma,\mathtt{TempID},\mathtt{Adr_{dst}})\leftarrow{\sf SendRequest}(gpk,sk,\mathtt{TempID},\allowbreak $

\noindent$\mathtt{Adr_{src}},\allowbreak \mathtt{Adr_{dst}},\mathtt{Adr_{proxy}})$, $(\sigma,\mathtt{TempID},\mathtt{Adr_{proxy}})\leftarrow{\sf RelayRequest}$

\noindent$(\mathtt{Adr_{src}},\allowbreak \mathtt{Adr_{dst}},\allowbreak \mathtt{Tbl},\sigma, \mathtt{TempID})$, and $C\leftarrow {\sf SendContent}(gpk,\sigma,\allowbreak \mathtt{TempID},\allowbreak M,\allowbreak \mathtt{Adr_{proxy}})$.
\end{definition}

\noindent 
Next, we define anonymity, semantic security, and unforgeability as follows. 
One session is defined as sequences of algorithm executions from ${\sf SendRequest}$ to ${\sf GetContent}$, where ${\sf SendRequest}\rightarrow {\sf RelayRequest} \rightarrow {\sf SendContent} \rightarrow \allowbreak{\sf RelayContent} \rightarrow {\sf GetContent}$. 
Anonymity guarantees that no adversary $\mathcal{A}$ who is allowed to communicate with the proxy (but is not allowed to know $\mathtt{Adr_{src}}$) can distinguish whether the users of two different sessions are the same or not. 
In this game, $\mathcal{A}$ is modeled as a malicious SP. 
Moreover, we care about signing key exposure, where $\mathcal{A}$ can obtain signing keys. In addition to this, we give $msk$ to $\mathcal{A}$ in order to guarantee that the KGC ability has nothing to right for identifying the user\footnote{Note that we exclude the trivial case that KGC is offered to generate a decryption key of $\mathtt{TempID}$ from a user whose IP address is $\mathtt{Adr_{src}}$, and observes that the transcript containing $\mathtt{TempID}$. }. 

\begin{definition}[Anonymity]~\label{def:anon}
\begin{enumerate}
\setlength{\itemsep}{0em}\setlength{\parsep}{0em}
\item The challenger $\mathcal{C}$ runs $(gpk,ik)\leftarrow{\sf GM.Setup}(1^\lambda)$ and $(params,msk)\leftarrow {\sf KGC.Setup}(1^\lambda)$, and computes two signing keys $sk_0,sk_1\leftarrow {\sf Join}(gpk,ik)$, 
and gives $gpk$, $sk_0$, $sk_1$, and $(params,msk)$ to an adversary $\mathcal{A}$. Moreover, $\mathcal{C}$ initializes $\mathtt{Tbl}:=\emptyset$. 

\item $\mathcal{A}$ is allowed to issue the ${\sf SendRequest}$ query $(i,\mathtt{TempID})\allowbreak \in\{0,1\}\times\mathcal{ID}$. $\mathcal{C}$ runs ${\sf SendRequest}(gpk,\allowbreak sk_b,\mathtt{TempID},\allowbreak \mathtt{Adr_{src}},\allowbreak \mathtt{Adr_{dst}},\mathtt{Adr_{proxy}})$, and returns $\sigma$ (generated via the ${\sf SendRequest}$ algorithm) to $\mathcal{A}$. 

\item $\mathcal{A}$ is allowed to issue the ${\sf RelayRequest}$ query $(\sigma,\mathtt{TempID})$. $\mathcal{C}$ runs ${\sf RelayRequest}(\mathtt{Adr_{src}},\allowbreak \mathtt{Adr_{dst}},\allowbreak \mathtt{Tbl},\allowbreak \sigma, \mathtt{TempID})$ and updates $\mathtt{Tbl}$.  

\item $\mathcal{A}$ is allowed to issue the ${\sf RelayContent}$ query $C$. $\mathcal{C}$ runs ${\sf RelayContent}(C,\allowbreak \mathtt{Tbl})$, and updates $\mathtt{Tbl}$. 

\item $\mathcal{A}$ sends $\mathtt{TempID}^\ast\in\mathcal{ID}$ to $\mathcal{C}$. 
$\mathcal{C}$ flips a coin $b\stackrel{\$}{\leftarrow}\{0,1\}$, and runs $(\sigma^\ast,\mathtt{TempID}^\ast,\allowbreak \mathtt{Adr_{dst}})\leftarrow{\sf SendRequest}(gpk,$

\noindent$sk_b,\mathtt{TempID}^\ast,\mathtt{Adr_{src}},\allowbreak \mathtt{Adr_{dst}},\allowbreak \mathtt{Adr_{proxy}})$ and $(\sigma^\ast,\allowbreak \mathtt{TempID}^\ast,\allowbreak \mathtt{Adr_{proxy}})\leftarrow{\sf RelayRequest}(\mathtt{Adr_{src}},\allowbreak \mathtt{Adr_{dst}},\mathtt{Tbl},\allowbreak \sigma^\ast, \mathtt{TempID}^\ast)$. 
$\mathcal{A}$ returns an arbitrary $C$ to $\mathcal{C}$. 
$\mathcal{C}$ runs $C\leftarrow{\sf RelayContent}(C,\allowbreak \mathtt{Tbl})$. 
Note that $\mathcal{A}$ can know the transcript of these algorithms executed by $\mathcal{C}$: $(\sigma^\ast,\mathtt{TempID}^\ast,\mathtt{Adr_{dst}})$, $(\sigma^\ast,\allowbreak \mathtt{TempID}^\ast,\allowbreak \mathtt{Adr_{proxy}})$, and $C$. 
$\mathcal{A}$ outputs $b^\prime\in\{0,1\}$. 
\end{enumerate}

\noindent
The protocol is said to have anonymity if ${\sf Adv}^{{\sf anon}}_{{\sf pro},\mathcal{A}}(\lambda):=|\Pr[b=b^\prime]-1/2|$ is negligible in $\lambda$.
\end{definition}

\noindent 
Next, we define semantic security which guarantees that no information of content $M$ is revealed from the transcripts of algorithms. In this game, an adversary $\mathcal{A}$ is modeled as a malicious proxy. Moreover, $\mathcal{A}$ is allowed to obtain $ik$ in order to guarantee that no information of $M$ is revealed even from the GM's viewpoint. 

\begin{definition}[Semantic Security]~
\begin{enumerate}
\setlength{\itemsep}{0em}\setlength{\parsep}{0em}
\item The challenger $\mathcal{C}$ runs $(gpk,ik)\leftarrow{\sf GM.Setup}(1^\lambda)$ and $(params,msk)\leftarrow {\sf KGC.Setup}(1^\lambda)$, and gives $gpk$, $ik$, and $params$ to an adversary $\mathcal{A}$. 

\item $\mathcal{A}$ is allowed to issue the ${\sf UserKeyGen}$ query $\mathtt{TempID}\in\mathcal{ID}$. $\mathcal{C}$ runs ${\sf UserKeyGen}(params,msk,\allowbreak \mathtt{TempID})$ and returns $dk_{\mathtt{TempID}}$. 

\item $\mathcal{A}$ sends $\mathtt{TempID}^\ast\in\mathcal{ID}, M_0^\ast,M_1^\ast\in\mathcal{M}$ and $sk^\ast$ to $\mathcal{C}$ as his choice, where $\mathtt{TempID}^\ast$ has not been sent as a ${\sf UserKeyGen}$ query. $\mathcal{C}$ flips a coin $b\stackrel{\$}{\leftarrow}\{0,1\}$, runs ${\sf SendRequest}(gpk,\allowbreak sk^\ast,\mathtt{TempID}^\ast,\allowbreak \mathtt{Adr_{src}},\mathtt{Adr_{dst}},\mathtt{Adr_{proxy}})$

\noindent and $C^\ast\leftarrow {\sf SendContent}(gpk,\sigma,\allowbreak \mathtt{TempID},\allowbreak M^\ast_b,\mathtt{Adr_{proxy}})$, and sends $(\sigma,\mathtt{TempID}^\ast,\mathtt{Adr_{dst}})$ and $C^\ast$ to $\mathcal{A}$. 

\item $\mathcal{A}$ is allowed to issue the ${\sf UserKeyGen}$ query $\mathtt{TempID}\in\mathcal{ID}$ where $\mathtt{TempID}\neq \mathtt{TempID}^\ast$. $\mathcal{C}$ runs ${\sf UserKeyGen}(params,\allowbreak msk,\allowbreak \mathtt{TempID})$ and returns $dk_{\mathtt{TempID}}$. 

\item Finally, $\mathcal{A}$ outputs $b^\prime\in\{0,1\}$. 
\end{enumerate}

\noindent 
The protocol is said to have semantic security if ${\sf Adv}^{{\sf ss}}_{{\sf pro},\mathcal{A}}(\lambda):=|\Pr[b=b^\prime]-1/2|$ is negligible in $\lambda$. 
\end{definition}

\noindent 
Finally, we define unforgeability which guarantees that no adversary $\mathcal{A}$ who does not have a signing key will be accepted by the ${\sf ValidityCheck}$ algorithm. In this game, $\mathcal{A}$ is modeled as a malicious user. Moreover, $\mathcal{A}$ is allowed to obtain $msk$ in order to guarantee that nobody can be accepted by SP even by KGC. 

\begin{definition}[Unforgeability]~
\begin{enumerate}
\setlength{\itemsep}{0em}\setlength{\parsep}{0em}
\item The challenger $\mathcal{C}$ runs $(gpk,ik)\leftarrow{\sf GM.Setup}(1^\lambda)$ and $(params,msk)\leftarrow {\sf KGC.Setup}(1^\lambda)$, and gives $gpk$, $params$, and $msk$ to an adversary $\mathcal{A}$. 
Moreover, $\mathcal{C}$ initializes $\mathcal{S}=\emptyset$. 

\item $\mathcal{A}$ is allowed to issue the ${\sf SendRequest}$ query $(i,\mathtt{TempID})$. If $sk_i$ has not been generated, then $\mathcal{C}$ runs $sk_i\leftarrow {\sf Join}(gpk,ik)$ and preserves $sk_i$. 
$\mathcal{C}$ runs ${\sf SendRequest}(gpk,\allowbreak sk_i,\allowbreak \mathtt{TempID},\allowbreak \mathtt{Adr_{src}},\allowbreak \mathtt{Adr_{dst}},\allowbreak \mathtt{Adr_{proxy}})$, and sends $\sigma$ to $\mathcal{A}$. 
Moreover, $\mathcal{C}$ appends $(\sigma,\mathtt{TempID})$ to $\mathcal{S}$. 

\item Finally, $\mathcal{A}$ outputs $(\sigma^\ast, \mathtt{TempID}^\ast)$. We say that $\mathcal{A}$ wins if $(\sigma^\ast, \mathtt{TempID}^\ast)\not\in\mathcal{S}$ and ${\sf ValidityCheck}(gpk,\sigma^\ast,\allowbreak \mathtt{TempID}^\ast)\allowbreak=1$. 
\end{enumerate}

\noindent 
The protocol is said to have unforgeability if ${\sf Adv}^{{\sf uf}}_{{\sf pro},\mathcal{A}}(\lambda):=\Pr[\mathcal{A}~\text{wins}]$ is negligible in $\lambda$. 
\end{definition}

\noindent 
We say that a protocol is called secure anonymous authentication protocol if the protocol is correct and has anonymity, semantic security, and unforgeability. 

%%%%%%%%%%%%%%%%%%%%%%%%%%%%%%%%%%%%%%%%
% Construction of Secure Anonymous Authentication Protocol
%%%%%%%%%%%%%%%%%%%%%%%%%%%%%%%%%%%%%%%%

\subsection{Protocol Construction}\label{Sec:ProtocolConstruction}

Here, we give our proposed protocol construction. 
First, we define the syntax of building blocks - IBE and open-free group signature - as follows:
%Note that their concrete security definitions are given in Appendix and Section~\ref{GS}, respectively.

An IBE scheme $\mathcal{IBE}$ consists of four algorithms: i.e., $({\sf IBE.Setup},\allowbreak {\sf Extract}, {\sf IBE.Enc}, {\sf IBE.Dec})$. 
Let $\mathcal{ID}$ and $\mathcal{M}$ be an identity space and message space, respectively. 

\begin{definition}[Syntax of IBE~\cite{[BonehF03]}]~
\begin{description}
\setlength{\itemsep}{0em}\setlength{\parsep}{0em}
\item[{\sf IBE.Setup}:] This algorithm takes as input the security parameter $\lambda$, and outputs a public key $params$ and a master secret key $msk$. 
\item[{\sf Extract}:] This algorithm takes as input $params$, $msk$, and an identity $ID\in\mathcal{ID}$, and outputs a decryption key $dk_{ID}$. 
\item[{\sf IBE.Enc}:] This algorithm takes as input $params$, $ID$, and a message $M\in\mathcal{M}$, and outputs a ciphertext $C_{IBE}$. 
\item[{\sf IBE.Dec}:] This algorithm takes as input $params$, $C_{IBE}$, and $dk_{ID}$, and outputs $M$. 
\end{description}
\end{definition}

\noindent 
We require the following correctness property: for all $(params,\allowbreak msk)\leftarrow{\sf IBE.Setup}(1^\lambda)$, all $ID$ and all $M$, $\Pr[{\sf IBE.Dec}(params,\allowbreak {\sf IBE.Enc}(params,ID,M), {\sf Extract}(params,msk,ID))=M]=1$ holds. 

An open-free group signature scheme $\mathcal{GS}$ consists of four algorithms: $({\sf GS.Setup}, {\sf Join}, {\sf Sign}, {\sf Verify})$\footnote{This new primitive is a kind of dynamic group signature, where a new member can join the system even after the setup phase. Note that, additional two algorithms, ${\sf Open}$ and ${\sf Judge}$, are usually contained in dynamic group signatures (e.g.~\cite{[BellareSZ05]}). 
%The ${\sf Open}$ algorithm identifies the actual signer by using the GM's secret key. 
The ${\sf Judge}$ algorithm checks a proof output by the ${\sf Open}$ algorithm, whether the ${\sf Open}$ algorithm is correctly executed or not. Obviously, the ${\sf Judge}$ algorithm is meaningless in the open-free variant.} as follows: 

\begin{definition}[Open-Free Group Signature]~
\begin{description}
\setlength{\itemsep}{0em}\setlength{\parsep}{0em}
\item[{\sf GS.Setup}:] This algorithm takes as input the security parameter $\lambda$, 
and outputs a group public key $gpk$ and an issuer key $ik$. 
\item[{\sf GS.Join}:] This algorithm takes as input $gpk$ and $ik$ (from GM), and a user is obtained a signing key $sk$. 

\item[{\sf Sign}:] This algorithm takes as input $gpk$, a signing key $sk$, and a message $M$, and outputs a group signature $\sigma$. 
\item[{\sf Verify}:] This algorithm takes as input $gpk$, $\sigma$, and $M$, and outputs $1$ if $\sigma$ is a valid signature on $M$, and 0 otherwise. 
\end{description}

%Note that we do not have to assume that the user has a secret key as a input of the ${\sf GS.Join}$ algorithm due to the open-free property. 
\end{definition}

\noindent 
We require the following correctness property: for all $(gpk,ik)\allowbreak \leftarrow {\sf GS.Setup}(1^\lambda)$ and $sk\leftarrow {\sf GS.Join}(gpk,ik)$, $\Pr[{\sf Verify}(gpk,\allowbreak {\sf Sign}(gpk,\allowbreak sk,M),M)=1]=1$ holds. 

Next, we give our proposed construction. In this construction, a signed message of the underlying group signature is $\mathtt{TempID}$ which is also regarded as a public key of the underlying IBE. 

\begin{con}[Proposed Protocol]~
\begin{description}
\setlength{\itemsep}{0em}\setlength{\parsep}{0em}
\item[${\sf GM.Setup}$]: Run $(gpk,ik)\leftarrow{\sf GS.Setup}(1^\lambda)$, and output $(gpk,\allowbreak ik)$. 

\item[${\sf KGC.Setup}$]: Run $(params,msk)\leftarrow{\sf IBE.Setup}(1^\lambda)$, and output $(params,msk)$. 

\item[${\sf Join}$]: Run $sk\leftarrow {\sf GS.Join}(gpk,ik)$, and output $sk$. 

\item[${\sf UserKeyGen}$]: Run $dk_{\mathtt{TempID}}\leftarrow{\sf Extract}(params,msk,\mathtt{TempID})$, 
and output $dk_{\mathtt{TempID}}$. 

\item[${\sf SendRequest}$]: Choose $\mathtt{TempID}\stackrel{\$}{\leftarrow} \mathcal{ID}$. Run $\sigma\leftarrow {\sf Sign}(gpk,sk,\allowbreak \mathtt{TempID})$, and send $(\sigma,\mathtt{TempID},\mathtt{Adr_{dst}})$ to the proxy whose IP address is $\mathtt{Adr_{proxy}}$. 

\item[${\sf RelayRequest}$]: Append $(\mathtt{TempID},\mathtt{Adr_{src}})$ to $\mathtt{Tbl}$, and relays a pair $(\sigma,\mathtt{TempID})$ and $\mathtt{Adr_{proxy}}$ to the destination SP whose IP address is $\mathtt{Adr_{dst}}$. 

\item[${\sf ValidityCheck}$]: Output $1$ if ${\sf Verify}(gpk,\sigma,\mathtt{TempID})=1$, and $0$, otherwise. 

\item[${\sf SendContent}$]: Output $\bot$ if ${\sf ValidityCheck}(gpk,\sigma,\mathtt{TempID})=0$. Otherwise, run $C_{IBE}\leftarrow {\sf IBE.Enc}(params,\mathtt{TempID},\allowbreak M)$, and send $C_{IBE}$ to the proxy whose IP address is $\mathtt{Adr_{proxy}}$. 

\item[${\sf RelayContent}$]: Relay $C_{IBE}$ to a user whose IP address is $\mathtt{Adr_{src}}$ contained in $\mathtt{Tbl}$. Moreover, remove $(\mathtt{TempID}, \allowbreak \mathtt{Adr_{src}})$ from $\mathtt{Tbl}$. 

%We assume that the proxy can decide the corresponding source IP address to be relayed by $C$. 

\item[${\sf GetContent}$]: Output the result of ${\sf IBE.Dec}(params,C_{IBE},\allowbreak dk_{\mathtt{TempID}} )$. 
\end{description}
\end{con}

Note that our the above construction only considers one-proxy setting, 
and therefore no anonymity is guaranteed from the viewpoint of the Proxy, 
since the Proxy directly relays communications between the User and SP. 
Note that this situation does not contradict our definition of anonymity (Def.~\ref{def:anon}). 
We can simply extend this protocol to a multi-proxy setting, where each Proxy relays $(\sigma, \mathtt{TempID})$ or $C_{IBE}$ between the previous Proxy and the next Proxy. Then, anonymity is guaranteed even from the Proxies' point of view unless all Proxies collude with each other. 

%%%%%%%%%%%%%%%%%%%%%%%%%%%%%%%%%%%%%%%%
% Performance enhancement
%%%%%%%%%%%%%%%%%%%%%%%%%%%%%%%%%%%%%%%%

\section{Group Signature}

The proposed secure anonymous authentication protocol uses a group signature scheme as its fundamental component. 
Though arbitrary group signature schemes could be used (i.e., by ignoring open functionality), it is beneficial to remove unnecessary functionality and improve performance efficiency, thus the proposed protocol in Section \ref{Sec:ProtocolConstruction} uses a group signature without open functionality. We call this an open-free group signature\footnote{A difference between ring signature~\cite{[RivestST01]} and open-free group signature is as follows. In ring signature schemes, a signer chooses a set of other members, and signs on behalf of the group of users. The anonymity of the signer cannot be revoked in contrast to group signature schemes. %(though sometimes ring signatures with linkability, e.g.,~\cite{[Fujisaki11]}, is considered to detect double signing for the same tag). 
However, a signer needs to involve/choose other members when the signer signs, and therefore needs to know other members. This does not match our setting. }. 
This section defines the security of such signatures. 

\subsection{Defining Open-free Group Signature}\label{GS}

In this section, 
we redefine the security definitions of the Furukawa-Imai group signature scheme, anonymity, traceability, and non-frameability, to match the open-free variant. 
Anonymity guarantees that no adversary $\mathcal{A}$ can distinguish whether two signers of group signatures are the same or not, even if $\mathcal{A}$ has the corresponding signing keys. 
Usually, there are two kind of anonymity, CPA-anonymity and CCA-anonymity. 
In CCA-anonymity, $\mathcal{A}$ is allowed to issue open queries, where $\mathcal{A}$ sends $(\sigma,M)$, and is given the result of the ${\sf Open}$ algorithm. 
Meanwhile, we do not have to consider these differences due to the open-free property. 

\begin{definition}[Anonymity]~
\begin{enumerate}
\setlength{\itemsep}{0em}\setlength{\parsep}{0em}
\item An adversary $\mathcal{A}$ with the security parameter $\lambda$ sends $gpk$, $sk_0$, $sk_1$, and $M$ to the challenger $\mathcal{C}$. 
\item $\mathcal{C}$ chooses $b\stackrel{\$}{\leftarrow}\{0,1\}$, 
computes $\sigma^\ast\leftarrow {\sf Sign}(gpk,sk_b,M)$, and sends $\sigma^\ast$ to $\mathcal{A}$. 
\item $\mathcal{A}$ outputs a bit $b^\prime\in\{0,1\}$. 
\end{enumerate}

\noindent 
An open-free group signature $\mathcal{GS}$ is said to have anonymity if ${\sf Adv}^{{\sf anon}}_{\mathcal{GS},\mathcal{A}}(\lambda):=|\Pr[b=b^\prime]-1/2|$ is negligible in $\lambda$. 
\end{definition}

Next, we redefine traceability. 
In usual definition, 
Traceability guarantees that no adversary $\mathcal{A}$ can produce a valid-but-untraceable group signature, that is, the ${\sf Open}$ algorithm cannot identify the corresponding signer though the ${\sf Verify}$ algorithm outputs 1. 
%In this game, $\mathcal{A}$ $\mathcal{A}$ can know the corresponding signing key. 
However, in the open-free variant, this definition is meaningless.
So, we define unforgeability here instead of traceability, where no adversary $\mathcal{A}$ can produce a valid group signature without knowing a signing key. 

\begin{definition}[Unforgeability]~
\begin{enumerate}
\setlength{\itemsep}{0em}\setlength{\parsep}{0em}
\item The challenger $\mathcal{C}$ runs $(gpk,ik)\leftarrow{\sf GS.Setup}(1^\lambda)$, and gives $gpk$ to an adversary $\mathcal{A}$. 
\item $\mathcal{A}$ is allowed to issue the signing query $(M,i)$. If a user $U_i$ has not been joined to the system, then $\mathcal{C}$ runs the ${\sf GS.Join}$ algorithm, computes $sk_i$, and returns $\sigma\leftarrow {\sf Sign}(gpk,sk_i,M)$ to $\mathcal{A}$. If $U_i$ has been joined to the system, then $\mathcal{C}$ returns $\sigma\leftarrow {\sf Sign}(gpk,sk_i,M)$ to $\mathcal{A}$. Moreover, $\mathcal{C}$ appends $(\sigma,M)$ into the list $\mathcal{S}$. 

\item Finally, $\mathcal{A}$ outputs $(\sigma^\ast,M^\ast)$.
We say that $\mathcal{A}$ wins if ${\sf Verify}(gpk,\sigma^\ast,M^\ast)=1$ holds and $(\sigma^\ast,M^\ast)\not\in\mathcal{S}$. 
\end{enumerate}

\noindent 
An open-free group signature $\mathcal{GS}$ is said to have unforgeability if ${\sf Adv}^{{\sf un}}_{\mathcal{GS},\mathcal{A}}(\lambda):=\Pr[\mathcal{A}~\text{wins}]$ is negligible in $\lambda$. 
\end{definition}

\noindent
%One may think that the above definition can be more stronger, where $\mathcal{A}$ is allowed to get a signing key (as in traceability of the conventional group signature), and $\mathcal{A}$ wins if $\mathcal{A}$ can produce a valid group signature using a signing key which is not given via the oracle. However, this cannot be well-defined, since nobody can distinguish whether $\mathcal{A}$ uses the own-made signing key or not, due to anonymity and the open-free property. 

Finally, we revisit non-frameability. 
Non-frameability guarantees that no adversary $\mathcal{A}$ can produce a valid group signature whose open result is an honest (i.e., uncorrupted by $\mathcal{A}$) user (say $U$). 
Obviously, this definition is meaningless in the open-free variant, and therefore we do not consider non-frameability. 

%Note that in non-frameability game, $\mathcal{A}$ can know all secret values, except $U$'s secret key (not exact signing key). 
Note that in order to achieve non-frameability in the original scheme, a user chooses a secret key $usk$, and is obtained its signing key $sk$ by executing the ${\sf GS.Join}$ algorithm with GM. What is critical, GM cannot know $usk$ itself (but can convince that the user knows $usk$ by using zero-knowledge proofs). In other words, we can remove a secret key $usk$ from the syntax of group signature unless non-frameability is required. 
This is is the reason why we do not require any secret key of users as input of the ${\sf GS.Join}$ algorithm, and the ${\sf GS.Join}$ algorithm can be a non-interactive algorithm. 

\subsection{Building Open-Free Group Signature}

Our group signature scheme modifies the Furukawa-Imai group signature~\cite{[FurukawaI06]}. 
In the Furukawa-Imai scheme, 
a user certificate issued by the GM is a short signature~\cite{[BonehB08]}. 
The user proves the possession of the certificate by NIZK proofs which are constructed via the Fiat-Shamir conversion~\cite{[FiatS86]}. 
For implementing the ${\sf Open}$ algorithm, an ElGamal-type double encryption is used over a decisional Diffie-Hellman (DDH)-hard group (in addition to bilinear groups).
In our open-free scheme, the DDH-hard group can be removed. 
%due to the open-free property. 
%Other part, NIZK proofs for a SDH pair, 
%\footnote{SDH stands for Strong Diffie-Hellman. Informally, we say that $q$-SDH assumption holds if no PPT adversary $\mathcal{A}$ can compute $(g_2^{1/(\gamma+x)},x)$ from $(g_1,g_2,g_2^\gamma,\ldots,g_2^{\gamma^q})$. }, 
Other part is the same as that of the original Furukawa-Imai group signature scheme. 

Note that, a simple construction, where for one signature verification/signing key pair $({\sf VK}, {\sf SK})$, each group member shares ${\sf SK}$, can also be seen as an open-free group signature scheme. However, this simple construction never realizes the revocation functionality~\cite{[LibertPY12]}. 
%Though we do not consider the revocation functionality in this paper, 
% (since our purpose is proof-of-concept implementation), 
Towards constructing a revocable open-free group signature scheme, we newly construct an open-free group signature scheme. 

\begin{con}[Proposed Open-Free Group Signature]~

\begin{description}
\setlength{\itemsep}{0em}\setlength{\parsep}{0em}
\item[{\sf GS.Setup}:] Let $(\mathbb{G}_1,\mathbb{G}_2,\mathbb{G}_T)$ be a bilinear group with prime order $p$, where $\langle g_1\rangle=\mathbb{G}_1$, $\langle g_2\rangle=\mathbb{G}_2$, and $e:\mathbb{G}_1\times\mathbb{G}_2\rightarrow \mathbb{G}_T$ be a bilinear map\footnote{We require bilinearity: for all $a,b\in\mathbb{Z}_p$, $e(g_1^a,g_2^b)=e(g_1,g_2)^{ab}=e(g_1^b,g_2^a)$, and non-degeneracy: $e(g_1,g_2)\neq 1_{\mathbb{G}_T}$, where $1_{\mathbb{G}_T}$ is the identity element in $\mathbb{G}_T$.}. 
Choose $\gamma\stackrel{\$}{\leftarrow}\mathbb{Z}_p$, and $h\stackrel{\$}{\leftarrow}\mathbb{G}_1$, and compute $W=g_2^\gamma$. Output $gpk=\linebreak(\mathbb{G}_1,\mathbb{G}_2,\mathbb{G}_T,e,g_1,g_2,h,W, e(g_1,g_2),e(g_1,W),H_3)$ and $ik\allowbreak =\gamma$, where $H_3:\{0,1\}^\ast\rightarrow \mathbb{Z}_p$ is a hash function modeled as a random oracle. 

\item[{\sf GS.Join}:] For a user $U_i$, choose $x_i,y_i\stackrel{\$}{\leftarrow}\mathbb{Z}_p$, compute $A_i=(g_1h^{-y_i})^{\frac{1}{\gamma+x_i}}$, and output $sk_i=(x_i,y_i,A_i)$. 

\item[{\sf Sign}:] Let $sk=(x,y,A)$. 
Choose $\beta\stackrel{\$}{\leftarrow}\mathbb{Z}_p$, set $\delta=\beta x-y$, and compute $T=A h^\beta$. Choose $r_{x},r_{\delta}, r_\beta\stackrel{\$}{\leftarrow}\mathbb{Z}_p$, and compute $R=e(h,g_2)^{r_\delta}e(h,W)^{r_\beta}/e(T,g_2)^{r_x}$, $c=H_3(gpk,T,R,M)$, $s_x=r_x+cx$, $s_\delta=r_\delta+c \delta$, and $s_\beta=r_\beta+c\beta$, and output $\sigma=(T,c,s_x,s_\delta,s_\beta)$. 

\item[{\sf Verify}:] Compute $R^\prime=\frac{e(h,g_2)^{s_\delta}e(h,W)^{s_\beta}}{e(T,g_2)^{s_x}}\big{(}\frac{e(T,W)}{e(g_1,g_2)}\big{)}^{-c}$, and output 1 if $c=H_3(gpk,T,R^\prime,M)$ holds, and 0 otherwise. 
\end{description}
\end{con}

Compared to the original Furukawa-Imai scheme, 
we can reduce three DDH-hard group elements and three $\mathbb{Z}_p$ elements. 
Accordingly, we can reduce the size of signature by 50\% compared to the original Furukawa-Imai group signature scheme. 

%Note that $e(A,g_2^x W)=e(g_1,g_2)e(h,g_2)^{-y}$ holds if $(x,y,A)$ is a valid certificate. From this equation, 
%$$\frac{e(T,W)}{e(g_1,g_2)}=\frac{e(h,g_2)^{\beta x-y}e(h,W)^\beta}{e(T,g_2)^x}$$ 
%holds for $T=A h^\beta$. Therefore, 

%\begin{gather*}
%\begin{split}
%\frac{e(h,g_2)^{s_\delta}e(h,W)^{s_\beta}}{e(T,g_2)^{s_x}}&=\frac{e(h,g_2)^{r_\delta}e(h,W)^{r_\beta}}{e(T,g_2)^{r_x}}\big{(}\frac{e(h,g_2)^{\beta x-y}e(h,W)^\beta}{e(T,g_2)^x}\big{)}^c\\
%&=R\big{(}\frac{e(T,W)}{e(g_1,g_2)}\big{)}^c
%\end{split}
%\end{gather*}
%holds. 

\section{Implementation and Discussion}

\subsection{Analysis on Proposed Protocol}

We can prove that our proposed protocol is secure if the underlying IBE scheme is IND-ID-CPA secure (like the Boneh-Franklin IBE scheme~\cite{[BonehF03]}) and the underlying group signature scheme is anonymous and unforgeable. 
We only give a sketch of proof of anonymity (other theorems can be similarly proved) here and omit the full proofs of the following theorems due to the page limitation. 

\begin{theorem}
Our protocol is anonymous if the underlying group signature scheme is anonymous. 
\end{theorem}

\begin{proof}[(Sketch)] Let $\mathcal{A}$ be an adversary who can break anonymity of our protocol. Then, we can construct an algorithm $\mathcal{B}$ that breaks anonymity of the underlying group signature scheme as follow. Let $\mathcal{C}$ be the challenger of the underlying group signature. 
$\mathcal{B}$ generates $gpk$, $sk_0$, and $sk_1$, and generates all IBE-related values. Then $\mathcal{B}$ gives $(gpk,sk_0,sk_1,\allowbreak params,\allowbreak msk)$ to $\mathcal{A}$. In the challenge phase, $\mathcal{B}$ gets $\mathtt{TempID}^\ast$ from $\mathcal{A}$, forwards it to $\mathcal{C}$, and gets $\sigma^\ast$ from $\mathcal{C}$. $\mathcal{B}$ uses $\sigma^\ast$ as the output of the ${\sf SendRequest}$ algorithm, and similarly simulates other algorithms. $\mathcal{A}$ outputs $b^\prime$ and $\mathcal{B}$ also outputs $b^\prime$ as the guessing bit. Then, $\mathcal{B}$ can break anonymity of the group signature with the same advantage of $\mathcal{A}$. This contradicts that the underlying group signature is anonymous.
\end{proof}

\begin{theorem}
Our protocol is semantic secure if the underlying IBE scheme is IND-ID-CPA secure. 
\end{theorem}

%\begin{proof}[(Sketch)]
%\end{proof}

\begin{theorem}
Our protocol is unforgeable if the underlying group signature scheme is unforgeable. 
\end{theorem}

%\begin{proof}[(Sketch)]
%\end{proof}

\subsection{Analysis on Group Signature}

The remaining part is to show that the proposed open-free group signature scheme is anonymous and unforgeable. 
The proposed open-free group signature scheme is constructed from an (honest-verifier) zero-knowledge proof of knowledge 
%of an SDH pair 
by using the Fiat-Shamir conversion~\cite{[FiatS86]}. First, we explain the original proof of knowledge protocol as follows. 
%(notations are the same as these of the proposed group signature). 
A prover computes $(T,R)$, and sends it to a verifier. The verifier sends a challenge value $c$ to the prover. 
The prover computes $(s_x,s_\delta,s_\beta)$, and sends it to the verifier. The verifier checks whether the verification equation 
%$R=\frac{e(h,g_2)^{s_\delta}e(h,W)^{s_\beta}}{e(T,g_2)^{s_x}}\big{(}\frac{e(T,W)}{e(g_1,g_2)}\big{)}^{-c}$ 
holds or not. 
Next, we show that this 3-move protocol is zero-knowledge (this immediately leads to anonymity). 
The simulator chooses $A\stackrel{\$}{\leftarrow}\mathbb{G}$ and $\beta\stackrel{\$}{\leftarrow}\mathbb{Z}_p$, and computes $T=A g_1^\beta$. 
Note that $\beta$ is chosen uniformly random. Therefore, $T$ generated from the simulator is drawn from a distribution that is indistinguishable from the distribution output by any particular prover. 
%\footnote{If the simulator needs to compute a ciphertext, then we need to use some complexity assumption. E.g., the decision linear assumption is required in the Boneh-Boyen-Shacham scheme since $T$ is a ciphertext of the linear encryption. }.
For $T\in\mathbb{G}$, the simulator chooses $c,s_x,s_\delta,s_\beta\stackrel{\$}{\leftarrow}\mathbb{Z}_p$, and computes $R=\frac{e(h,g_2)^{s_\delta}e(h,W)^{s_\beta}}{e(T,g_2)^{s_x}}\big{(}\frac{e(T,W)}{e(g_1,g_2)}\big{)}^{-c}$. Then the transcript $(T,R,c,\allowbreak s_x,s_\delta,s_\beta)$ here is indistinguishable from transcripts of the actual protocol. 

Next, we show that the protocol is a proof of knowledge. That is, we show there exists an extractor that can extract a SDH pair from $(T,R,c,s_x,s_\delta,s_\beta)$ and $(T,R,c^\prime,s^\prime_x,s^\prime_\delta,s^\prime_\beta)$, where $c\neq c^\prime$ and both transcripts satisfy the verification equation. 
Set $\tilde{x}:=\frac{s_x-s^\prime_x}{c-c^\prime}$, $\tilde{y}:=\frac{(s_x-s^\prime_x)(s_\beta-s^\prime_\beta)-(s_\delta-s^\prime_\delta)(c-c^\prime)}{(c-c^\prime)^2}$, and $\tilde{\beta}:=\frac{s_\beta-s^\prime_\beta}{c-c^\prime}$. 
Then, $\frac{e(T,W)}{e(g_1,g_2)}=\frac{e(h,g_2)^{\tilde{\beta}\tilde{x}-\tilde{y}}e(h,W)^{\tilde{\beta}}}{e(T,g_2)^{\tilde{x}}}$ holds. Therefore, for $\tilde{A}=T/h^{\tilde{\beta}}$, $e(\tilde{A},g^{\tilde{x}}W)=e(g_1,g_2)e(h,g_2)^{-\tilde{y}}$ holds. That is, $(\tilde{x},\tilde{y},\tilde{A})$ can be extracted. 
This immediately leads to unforgeability. 
%Briefly, by the above extractor and the Forking Lemma~\cite{[PointchevalS00]}, the simulator rewinds the adversary, obtains two forged signatures, and can extract a SDH pair from forged signatures. 
We omit the formal proof since this is similar as that of the original Furukawa-Imai scheme.

\subsection{Prototype}

This section introduces a prototype that implements the proposed protocol and evaluates its performance to demonstrate the feasibility and practicality of the protocol.

\subsubsection{Implementation}

%The prototype consists of the User, Proxy, and SP modules.
We built the User and SP modules by using C language (GCC version 4.2.1). 
We also used the TEPLA library~\cite{TEPLA} for implementing the Boneh-Franklin IBE scheme and our open-free group signature scheme. 
This library supports optimal Ate pairings over Barreto-Naehrig (BN) elliptic curves~\cite{[BarretoN05]} with 254-bit prime order and the corresponding embedded degree is 12. 
This enables 128-bit security.
We used Simpleproxy~\cite{Simpleproxy} for the Proxy module.

Three types of communication sequences are implemented, i.e., User-GM, User-KGC, and User-Proxy-SP, and each of the sequence runs the modules defined in Section \ref{Sec:ProtocolConstruction}.
The User-GM sequence begins with the Join module, which communicates with the GM. The GM then computes the signing key $sk$, and returns it to the User.
The User-KGC sequence begins with the UserKeyGen module, which communicates with the KGC. The User sends $\mathtt{TempID}$ to the KGC, and the KGC then computes  the decryption key $dk_{\mathtt{TempID}}$, and returns it to the User.
The User-Proxy-SP sequence begins with the SendRequest module that sends a group signature and TempID.
Upon receiving them, the Proxy runs RelayRequest module that forwards them to the SP.
It then runs ValidityCheck module and SendContent module that returns an IBE ciphertext to the Proxy, which forwards that to the User.
The User then runs GetContent module that decrypts the IBE ciphertext by using the corresponding $dk_{\mathtt{TempID}}$. 

Note that the User-GM sequence needs to be run before User-Proxy-SP sequence starts. Likewise, the User-KGC and User-Proxy-SP sequences are run in parallel, though the User-KGC procedure needs to be completed before User-Proxy-SP procedure's ${\sf Get Content}$ module is run. 
%Proxy registers each pair $(\mathtt{TempID}, \mathtt{Adr_{src}})$ in a table. 

\begin{figure}[htb]
\centering
\includegraphics[scale=.31]{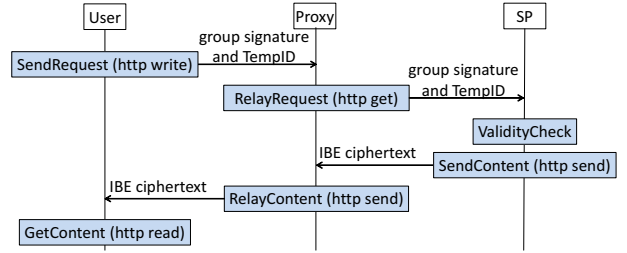}
\caption{Sequences for User-Proxy-SP}\label{seq123}
\end{figure}

\subsubsection{Performance Measurement}

An environment for performance evaluation of the proposed protocol was prepared.
We used an Apple MacBookPro 15inch mid 2010 (processor: 2.8GHz Intel Core i7, Memory: 8GB, 1067 MHz DDR3, Darwin Kernel Version 12.4.0), and prepared two VMs by using VMware Fusion 5.0.3. 
We assigned the roles of the User to the MacOS the roles of the Proxy and SP on the VMs. For the Proxy, the VM ran FreeBSD amd64 9.1-RELEASE with one processor and 256MB of memory, and for the SP, VM ran CentOS 5.9 x86\_64 with one processor and 512MB of memory.

Here, we show that our protocol is feasible by showing the running time of algorithms and total running time of one session are msec order.
First, we show the running time of one session (User$\rightarrow$Proxy$\rightarrow$SP$\rightarrow$Proxy$\rightarrow$User) in the following cases: (1) HTTP communications (i.e., without any cryptographic operations), (2) SSL communications, and (3) our protocol in Table~\ref{session-time}.
To measure the running time of the SSL communication, we use the s\_server/s\_time command of the OpenSSL library (ver. 1.0.1e)~\cite{openssl}.
We use DHE-RSA-AES128-SHA256 cipher suite with a 3072-bit size public key since this also supports 128-bit security, as in ours.

\begin{table}[h]
\centering
\caption{Running Time (one session)}\label{session-time}
\begin{tabular}{|c|c|c|} 
\hline
Scheme & Time(msec) & Cryptographic Operations\\ \hline\hline
None & 4.714  & - \\ \hline
SSL & 12.897  & Enc/Auth\\ \hline
Ours & 624.743 & Enc/Anon. Auth\\\hline
\end{tabular}
\end{table}

\noindent 
Table \ref{session-time} shows that the running time of our protocol is approximately 50-times slower than that of SSL communications. 
This inefficiency is due to the pairing computation which is not required in usual public key encryption, digital signature, and authentication (these are used in SSL). 
Nevertheless, 
it is particularly worth noting that our running time still fits inside millisecond order, 
and our protocol even supports secure, anonymous, and authenticated communication, simultaneously. 

For reference, Table \ref{exp} gives the running time of each algorithm.
Note that the ${\sf GM.Setup}$, ${\sf KGC.Setup}$, and ${\sf Join}$ algorithms can be run offline, and the ${\sf UserKeyGen}$ algorithm can be run separately against the session.
Moreover, we ignore the ${\sf RelayRequest}$ and ${\sf RelayContent}$ algorithms since these (run by Proxy) just relay the communication, and are run independently against any cryptographic operations.

\begin{table}[h]
\begin{center}
\caption{Running Time (algorithms)}\label{exp}
\begin{tabular}{|c|c|c|} 
\hline
Algorithm & Time(msec) & Entity \\ \hline\hline
${\sf GM.Setup}$ & 105.712 &  GM\\ \hline
${\sf KGC.Setup}$ & 102.883 &  KGC\\ \hline
${\sf Join}$ & 109.036 & User-GM\\\hline
${\sf UserKeyGen}$ & 102.958  &  User-KGC\\ \hline
${\sf SendRequest}$ & 125.069  &  User\\ \hline
%${\sf RelayRequest}$  &  &  Proxy\\ \hline
${\sf ValidityCheck}$ &199.247  &  SP\\ \hline
${\sf SendContent}$  & 198.636 &  SP\\ \hline
%${\sf RelayContent}$  &  &  Proxy\\ \hline
${\sf GetContent}$  & 101.158 &  User\\ \hline
\end{tabular}
\end{center}
\end{table}

\noindent 
The dominant factor for User is the ${\sf SendRequest}$ algorithm which computes a group signature. 
Note that this procedure can also be run offline by assuming that the User chooses $\mathtt{TempID}$ and computes a group signature before starting a session. 
Then, the total running time of one session becomes less than 500 msec.

\section{Conclusion}

The proposed protocol along with our group signature enables secure anonymous authentication.
It is feasible and practical in terms of transaction time.
Although this paper proved its concept, we need to consider practical deployment over the Internet.
Indeed, the protocol requires a proxy that assists secure, anonymous, and authenticated communication.
Though various types of proxy may exist, including Tor routers, we need to consider and verify the adaptability of our protocol to the current infrastructure.
On the other hand, assorted anonymous communication systems~\cite{TorProject,oakland2013-parrot,[MoghaddamLDG12]} have risks of being used by malicious parties.
One reason for that is their inability to authenticate users. Properly applying our protocol may enable these systems to be properly used.
Through this work, we wish to facilitate secure, anonymous, and authenticated communication over the Internet.

\noindent$\mathbf{Acknowledgment}$: 
The authors would like to thank Dr. Goichiro Hanaoka and Dr. Miyako Ohkubo for their invaluable comments.

\end{document}